\documentclass[letterpaper, 10 pt, conference]{ieeeconf}
\usepackage{graphicx,charter, latexsym}
\newtheorem{definition}{Definition}

\usepackage{amssymb, amsmath,graphicx,charter, latexsym}

\newtheorem{theorem}{Theorem}

\usepackage{multirow,url,float,epstopdf,cite}
\usepackage{verbatim}
\newtheorem{assumption}{Assumption}
\usepackage{algorithmic}
\usepackage{algorithm}
\usepackage{tikz}
\usepackage{pifont}
\usepackage{xcolor}
\usetikzlibrary{patterns,calc}
\usetikzlibrary{shapes.geometric, arrows}
\usepackage[overlay,absolute]{textpos}

\tikzstyle{block}=[draw opacity=0.7,line width=1.4cm]
\tikzstyle{ag} = [circle, radius =3cm, text centered, draw=black]
\tikzstyle{startstop} = [rectangle, rounded corners, minimum width=3cm, minimum height=1cm,text centered, draw=black, fill=red!30]
\tikzstyle{io} = [trapezium, trapezium left angle=70, trapezium right angle=110, minimum width=1cm, minimum height=1cm, text centered, draw=black, fill=blue!30]
\tikzstyle{process} = [rectangle, draw,fill=orange!30, text width = 20em, text centered, rounded corners, minimum height=4em, minimum width=1cm]
\tikzstyle{upd} = [rectangle, draw,fill=orange!30, text width = 5em, text centered, rounded corners, minimum height=4em, minimum width=1cm]
\tikzstyle{decision} = [diamond, minimum width=3cm, minimum height=1cm, text centered, draw=black, fill=green!30]
\tikzstyle{arrow} = [thick,->,>=stealth]
\tikzset{
  solid node/.style={circle,draw,inner sep=1.2,fill=black},
  hollow node/.style={circle,draw,inner sep=1.2},
}

\IEEEoverridecommandlockouts

\begin{document}
\title{The ISO Problem: Decentralized Stochastic Control via Bidding Schemes}
\author{Rahul Singh, P. R. Kumar and Le Xie 
\thanks{This material is based upon work partially supported by NSF under Contract Nos. ECCS-1546682, CPS-1239116, and Science \& Technology Center Grant CCF-0939370.}
\thanks{Rahul Singh, P. R. Kumar and Le Xie are with Texas A\&M University, College Station, TX 77840, USA.
        {\tt\small rsingh1@tamu.edu, prk@tamu.edu, lxie@tamu.edu}}%
}

\maketitle
\IEEEpeerreviewmaketitle
\begin{abstract}
We consider a smart-grid connecting several agents, modeled as stochastic dynamical systems, who may be electricity consumers/producers. At each discrete time instant, which may represent a $15$ minute interval, each agent may consume/generate some quantity of electrical energy. The Independent System Operator (ISO) is given the task of assigning consumptions/generations to the agents so as to maximize the sum of the utilities accrued to the agents, subject to the constraint that energy generation equals consumption at each time.

This task of coordinating generation and demand has to be accomplished by the ISO without the agents revealing their system states, dynamics, or utility/cost functions. We show how and when a simple iterative procedure converges to the optimal solution. The ISO iteratively obtains electricity bids by the agents, and declares the tentative market clearing prices. In response to these prices, the agents submit new bids.

On the demand side, the solution yields an optimal demand response for dynamic and stochastic loads. On the generation side, it provides the optimal utilization of stochastically varying renewables such as solar/wind, and generation with fossil fuel based generation with dynamic constraints such as ramping rates. Thereby we solve a decentralized stochastic control problem, without agents sharing any information about their system models, states or utility functions.
\end{abstract}

\IEEEpeerreviewmaketitle

\section{Introduction}
We consider the problem faced by the electricity grid operator, called the Independent System Operator (ISO). If the ISO knows the total demand of the loads, it has to solve the problem of allocating the required power among different generators so that the total cost of production is minimized\footnote{There are additional aspects such as security against contingencies, etc., that we neglect here.}. This problem can be solved by the generators bidding their marginal cost curves, and the ISO performing the optimization to obtain and declare the market clearing price. There are also additional aspects such as assuring that the power flow can be delivered over the network~\cite{bose1,oren1,gali1}.

The above deterministic static model with a fixed demand is insufficient for the oncoming era when we want to maximize the integration of renewable energy sources such as wind and photo voltaic, which are dynamic and vary unpredictably with time. Their modeling requires a dynamic stochastic system. Dynamic models can also be used to model features such as ramping constraints that are important for modeling fossil fuel generators. 

When employing renewable energy we need to make the demand adapt to the availability of renewable energy, called ``demand response", in contrast to the traditional scenario where demand is inflexible and supply needs to match whatever demand is. Loads generally have dynamic constraints, since, for example, air conditioners can be deferred for a while but not indefinitely. So loads also need to be modeled as dynamic control systems. Further, since environmental variables such as temperature are involved, they may be uncertain, and hence will also generally need to modeled as stochastic dynamic systems.

Such dynamic models can also model storage devices where the state is the amount of energy stored, and prosumers, such as homes with solar panels, which may switch at uncertain times from being consumers to generators. Therefore we model all the agents involved as stochastic dynamical systems.  

The goal in operating this system is to maximize the sum of the utilities of all the agents. There are however several constraints on information sharing that need to be respected. Individual agents may be averse to sharing system states with each other. They may not even be willing to share their individual system models or their individual utility functions for several reasons ranging from the competitive nature of commercial enterprises, to protecting privacy of consumers' home states.

The overall systemwide optimality of such a system is sought to be achieved by the ISO. It needs to both determine the optimal demand response over time, as well as allocate generation over time among the lowest cost generators, in the face of stochastic uncertainty, and to do so at minimum systemwide cost. The ISO would like to achieve this by simply determining prices and leaving each agent to its own selfish utility maximization, as in general equilibrium theory~\cite{arrow}.

Our main results are the following. We establish iterative interaction processes such as tatonnement~\cite{walras1} under which the ISO can indeed perform this task for certain stochastic dynamic systems. We also address the complexity of this task under several scenarios. In the case where the agents can be modeled as linear Gaussian stochastic systems and the cost functions are quadratic, we show that a simple scheme not involving contingent markets yields the systemwide global optimum.
\section{System Model}\label{systemmodel}
Consider a smart-grid consisting of $M$ agents which may be generators, loads, prosumers or storage devices, each modeled as a stochastic dynamical system. The following are the key ingredients of our system:
\begin{enumerate}
\item \textit{Randomness } is modeled through a probability space $\left(\Omega,\mathcal{F},\mathbb{P} \right)$. The ``state of the world" $\omega\in \Omega$, and captures ``random" phenomena such as weather, wind-speed, coal shortage, or a damaged wind-turbine. It affects agent $i$ through the random processes $N_i(\omega,t)$ and $N_c(\omega,t)$ $t=0,1,2,\ldots,T-1$. We regard $N_c(t)$ as a ``common" uncertainty that affects and is known to causally by all agents (e.g., temperature of a city), while
$N_i(t)$ is a ``private" uncertainty specific and known causally only to agent $i$. We denote by $\mathcal{F}_t$ the sigma algebra generated by all the noises up to and including time $t$.
\item \textit{Agents: } Each agent $i$ has a state $X_i(t)$ at time $t$ known to it, that evolves as,
\begin{align}\label{dyna1}
X_i(t+1) = f_i^t(X_i(t),U_i(t),N_i(t),N_c(t)),
\end{align}
where $U_i(t)\in \mathbb{R}$ is the amount of electricity supplied (negative if consumed) to the grid by agent $i$ at time $t$. Each discrete time instant corresponds to a $15$-minute interval of the real-time market implemented by the ISO.
\item One-step \textit{Cost function} of an agent $i$, $c_i(x_i,u_i)$ (or its negative, a one-step utility function $-c_i(x_i,u_i)$). For producers, this cost could be due to labor, coal, etc.. For consumers, this could represent the cost incurred due to the high temperature of house/business facility, or the cost incurred due to a delay in performing a task resulting from non-purchase of electricity.
\item \textit{System Operating Cost} is  the expected value of the sum of the finite horizon total costs incurred over the time duration $\{0,1,2,\ldots,T\}$ by all the agents,
\begin{align}\label{swf}
\mathbb{E}\left(\sum_{t=0}^{T} \sum_{i=1}^{N}c_i \left(X_i(t),U_i(t)\right) \right).
\end{align}
The time horizon $T$ can, for example, be $96$ which corresponds to one day. It is the total electricity generation cost minus the utility provided to the consumers. 
\item \textit{Energy Balance Constraint}: The basic constraint that we will focus on is that total generation must equal total consumption at each time $t$, i.e., $\sum_{i=1}^N U_i(t) = 0$ at each time $t$.
\item \textit{The ISO} is an agency that accepts electricity purchase/sale bids that are submitted by the agents for each time slot $t=0,1,\ldots,T$. In our model, we allow for the agents and ISO to iterate on the bids before the market clearing price is declared. Once the iterations have converged, the ISO declares the market clearing prices, and the agents generate/consume the agreed electrical energies at the declared prices.  
\end{enumerate}
\section{The ISO Problem}\label{isoproblem} 
The ISO problem is to solve the following constrained stochastic dynamic control problem,
\begin{align}\label{p0}
& \min \mathbb{E} \left\{\sum_{t=0}^{T} \sum_{i=1}^{M}c_i \left(X_i(t),U_i(t)\right)\right\}\notag\\
&\mbox{such that } \sum_i U_i(t) = 0, t=0,1,\ldots,T-1, \notag\\
&\mbox{and } X_i(t+1) = f_i^t(X_i(t),U_i(t),N_i(t),N_c(t)), \mbox{ for }\notag\\
& i =1,2,\ldots,M,\mbox{ and } t=0,1,\ldots,T-1.
\tag{ISO Problem}
\end{align} 
The expectation above is taken with respect to the combined uncertainty or ``noise" process $N(t):=(N_1(t),N_2(t),\ldots,N_M(t),N_c(t))$ for $t=0,1,\ldots,T-1$.

The noises $\left\{N_1(t), N_2(t),\ldots,N_M(t),N_c(t)\right\}$ may be dependent random variables. There can also be dependence across time.
\section{Problem Statement, Key Questions and Goals}\label{gol}
The key issues of interest are the following:
\begin{enumerate}
\item Since the ISO Problem is a decentralized control problem with non-classical information structurel~\cite{wsh,sardar}, we would like to know whether it is possible to achieve the exact optimal performance as attained by centralized control. As our analysis will show, the ISO is able to attain optimal coordination amongst the $M$ dynamic systems through announced ``prices" under some conditions. 
\item \textit{Decentralized Optimization and Sufficient Statistics} The~\ref{p0} can be viewed as a constrained Markov Decision Process (MDP)~\cite{altman}. The current state of our knowledge does not allow us to handle general MDPs with different observation patterns and different cost functions. Our scheme establishes certain ``sufficient statistics" for each agent to make optimal decisions in a decentralized environment~\cite{striebel,blackwell}.
\end{enumerate}
We will show that there exist simple ``iterative bidding schemes" (IBS) which yield the same performance as that of the optimal centralized controller under some models.
\section{Related Works}
No similar results appear to be known for the general decentralized stochastic control problem. Team problems have been extensively studied, e.g.,~\cite{radnor,sardar}, but the formulations are very restrictive in that each agent needs to know the system dynamics of the other agents. Even when the models are known, there are still considerable difficulties in decentralized stochastic control. When agents do not share observations, severe complexity can set in, even in an otherwise linear quadratic Gaussian problem, as pointed out by Witsenhausen in his counterexample of a two stage problem~\cite{wsh}. The role of observation, signaling~\cite{sardar}, and the trade-off between communication and control, are evident from this counterexample~\cite{wsh}. There are some recent structural results~\cite{nayar} and on sufficient statistics~\cite{jeff} under the restrictive assumption that the agents know the dynamics of all other agents. But the proposed solutions suffer from the curse of dimensionality.

We show below that the ISO Problem is an example of a decentralized control system with non-classical observation patterns in which signaling can successfully result in globally optimum performance. The agents need not reveal their observations, state values, their individual system dynamics or their individual cost-functions. We construct concrete signaling schemes which encode-decode the information required in order to recover the same performance as that of centralized control. From the economics side, this work is an extension of general equilibrium theory~\cite{arrow1}. To the authors' knowledge there does not appear to be any similar result for coordinating multiple LQG systems or the efficiency of the simplified signaling. 

In the setting of the energy market,~\cite{schwepe,schwepe1} discuss a general framework for the operation of an electicity grid via the establishment of a marketplace in which the electricity purchases are done on the basis of spot prices. Thus, electrical energy is treated as a commodity which can be bought, sold, and traded, taking into account its time-and space-varying values and costs, and a framework is presented for the establishment of an energy marketplace. Our results can be viewed as providing guidelines for the operation of this marketplace via simple bidding schemes when the agents are nonlinear stochastic dynamical systems. Our scheme maximize the net social utlity.

Viewed from the power system end, there have been many efforts since the deregulation of the electricity sector on a market-based framework to clear the system. Ilic et al. \cite{Ilic2011framework} have proposed a two-layered approach that internalizes individual constraints of market participants while allowing the ISO to manage the spatial complexity. The approximated MPC algorithm is shown to perform well in many realistic applications.
 
In order to analyze the strategic interactions between the ISO and market participants, game theoretical approaches have been proposed, e.g., Zhu et al\cite{Zhu2013} uses a Stackelberg game framework for studying economic dispatch with demand response. Wang et al ~\cite{Wang2014gametheoretic} pose the problem as a noncooperative game, and prove the existence of a Nash equilibrium under some assumtpions. 

One of the major challenges we address is how to elicit optimal demand response without revealing the inherent dynamic models of the loads to the ISO.
\section{A Tree Visualization of System Randomness }\label{sec:tree}
A tree visualization of the system randomness will be insightful in the discussions to follow; see Fig~\ref{extenform}. The combined system comprising of all the $M$ agents evolves as,
\begin{align}
X(t+1) = f^t(X(t),U(t),N(t)).
\end{align}
Let us assume for the time being that the noise process $N(t)$ is allowed to assume only finitely many values at each time. We can then construct an uncertainty tree of depth $T$, in which the root node corresponds to initial system state, and every path from the root to a leaf node corresponds to a unique realization of the noise sequence $\left(N(0),N(1),\ldots,N(T-1)\right)$.
\begin{figure}[h]\resizebox{8.5cm}{5cm}{ 
\begin{tikzpicture}[node distance=4cm][font=\scriptsize]
  \tikzset{
    level 1/.style={level distance=15mm,sibling distance=90mm},
    level 2/.style={level distance=15mm,sibling distance=40mm},
    level 3/.style={level distance=15mm,sibling distance=20mm},
    level 4/.style={level distance=15mm,sibling distance=20mm},
  }

  \node[hollow node,label=above:{Initial State $\left(x_1,x_2\right)$}]{}
    child{node[solid node,label=left:{}]{}
      child{node(l1)[solid node]{}
        child{node[solid node,label=below:{}]{}edge from parent node[left]{$N_1(2)=0$}}
        child{node[solid node,label=below:{}]{}edge from parent node[right]{$N_1(2)=1$}}
        edge from parent node[left]{$N_2(1) =0$}
      }
      child{node(l2)[solid node]{}
        child{node[solid node,label=below:{}]{}edge from parent node[left]{$N_1(2)=0$}}
        child{node[solid node,label=below:{}]{}edge from parent node[right]{$N_1(2)=1$}}
        edge from parent node[right]{$N_2(1) =1$}
      }
      edge from parent node[left,xshift=-10]{$N_1(0)=0$}
    }
    child{node[solid node,label=right:{}]{}
      child{node(r1)[solid node]{}
        child{node[solid node,label=below:{}]{}edge from parent node[left]{$N_1(2)=0$}}
        child{node[solid node,label=below:{}]{}edge from parent node[right]{$N_1(2)=1$}}
        edge from parent node[left]{$N_2(1) =0$}
      }
      child{node(r2)[solid node]{}
        child{node[solid node,label=below:{}]{}edge from parent node[left]{$N_1(2)=0$}}
        child{node[solid node,label=below:{}]{}edge from parent node[right]{$N_1(2)=1$}}
        edge from parent node[right]{$N_2(1) =1$}
      }
      edge from parent node[right,xshift=10]{$N_1(0) =1$}
    }
;
\end{tikzpicture}}
\caption{A Tree based visualization of randomness for a two agent system evolving over three bid times. The noise values are allowed to be binary, $0$ or $1$.}
\label{extenform}
\end{figure}
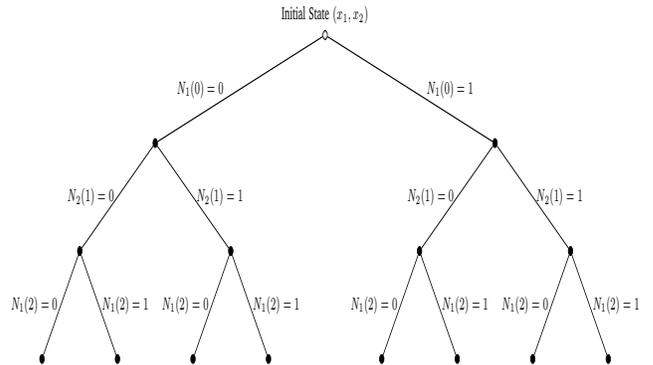
\section{Iterative Bidding Schemes}\label{sec:ibs}
The key contribution of this work is to propose solutions to the~\ref{p0} in the form of Iterative Bidding Schemes (IBS), as in Walrasian tatonnement~\cite{arrow}. Such schemes intertwine two simple processes, which we call \textit{Bid Update} and \textit{Price Update}. We begin by defining two key elements of the IBS, the bid function and the price function.

\textit{Bid Function}: A bid \emph{sequence} by agent $i$ specifies to the ISO how much electricity that agent will purchase (negative if supplying) in every time period from that time till the final time. At time $t$ it is a sequence of the form $(U_i(t), U_i(t+1),\ldots, U_i(T))$. A bid \emph{function} (in short just ``bid") specifies the bid sequence, as a function of the past history of observed noise $N(s), s<t$. In Figure~\ref{extenform} the bid function of each agent simply specifies, for each node in the tree, the amount of electricity that agent $i$ is willing to purchase when the system passes through that node if it ever does so. We note that $U_i(\omega,t)$ is adapted to the filtration $\mathcal{F}_t$. The bid function of agent $i$ will be denoted by $U_i$.

A \textit{price function} is a function announced by the ISO, which specifies for each time $t$, as a function of the past history of noise $N(s), s<t$, the price $\lambda(t)$ at which electricity will be sold/bought in the market. In the tree example of Figure~\ref{extenform}, this corresponds to the market clearing price corresponding to each node of the tree. The price function $\{\lambda(\omega,t)\}$ is also an $\mathcal{F}_t$-adapted stochastic process, which will be denoted by $\lambda$.

\textit{Bid Update}: Let us suppose, for the  time being, that the ISO has somehow declared a price function $\lambda$. In the \textit{Bid Update}, each agent $i$ changes its bid in response to the price function $\lambda$. In order to derive its new bid, it solves the following problem, dubbed Agent $i$'s Problem,
\begin{align*}
\min \mathbb{E}\left\{\sum_{t=0}^{T} c_i(X_i(t),U_i(t)) + \lambda(t) U_i(t) \right\}.
\end{align*}
It maximizes agent $i$'s total net utility, defined as the utility $-c_i(X_i(t))$ it derives from its state being $X_i(t)$, minus the amount $\lambda(t)U_i(t)$ it pays for the electricity.

\textit{Price Update} The ISO updates the price function in response to the agents having submitted their updated bids. Since in our context here the sole purpose of the ISO is to make sure that the net demand equals net supply, we will consider a simple rule by which it raises prices if demand exceeds supply and reduces otherwise, i.e., guided by the excess consumption function. Suppose the previous price was $\lambda^k$ and the bid was $U^k$. Then, the Price Update is,
\begin{align*}
\lambda^{k+1}(t) = \lambda^{k}(t) \left(1-\alpha_k\right) + \alpha_k \left(\sum_i U^k_i(t)\right)
\end{align*}
where $\alpha_k > 0$ is an ``adaptation gain". We employ the choice $\alpha_k = 1/k$, which satisfies the twin conditions $\sum_{k=0}^\infty \alpha_k = \infty$, and $\sum_{k=0}^\infty \alpha_k^2 < +\infty$, a common convergence condition in stochastic approximation~\cite{borkarbook}.

It will be the object of the following section, to show that an iteration of Bid Update-Price Update can solve the~\ref{p0} under some conditions.
\section{The Deterministic Case}\label{cdc}
First we consider the ISO Problem for deterministic systems,
\begin{align}\label{detp}
& \min \sum_{t=0}^{T} \sum_{i=1}^{M}c_i \left(x_i(t),u_i(t)\right)\notag\\
&\mbox{such that } \sum_i u_i(t) = 0,\mbox{ for } t=0,1,\ldots,T, \notag\\
&\mbox{and } x_i(t+1) = f_i^t(x_i(t),u_i(t)), \mbox{ for }\notag\\
& i =1,2,\ldots,M,\mbox{ and } t=0,1,\ldots,T-1.
\end{align} 
The intermediate variables $x_i(t)$ can be expressed in terms of the inputs $u_i:=\left(u_i(1),u_i(2),\ldots,u_i(T-1)\right)$
 and thus the cost term $\sum_{t=1}^{T} \sum_{i=1}^{M}c_i \left(x_i(t),u_i(t)\right)$ can also be expressed solely as a function of the inputs $u_i, i=1,2\ldots,M$. Convexity plays a major role, as noted by Arrow~\cite{arrow}.
\begin{assumption}[Convexity Assumption]\label{convexass}
For $i=1,2,\ldots,M$, the function $\sum_{t=1}^{T} c_i \left(x_i(t),u_i(t)\right)$ is convex in the input vector $\left(u_i(1),u_i(2),\ldots,u_i(T-1) \right)$. 
\end{assumption} 
We will now derive a solution to the~\ref{p0} under Assumption~\ref{convexass}, and show that it achieves the same performance as that of optimal centralized control.

The cost is convex in the vectors $u_i$. Employing the definition of each $x_i(t)$ as $f^t_i\left(x_i(t-1), u_i(t-1)\right)$, the associated Lagrangian and dual function are given by,
\begin{align*}  
\mathcal{L}\left(u, \lambda \right): &= \sum_{i=1}^{M}\left\{\sum_{t=0}^{T} c_i(x_i(t)) + \lambda (t) u_i(t)\right\},\\
D(\lambda) :&= \min_{u}  \mathcal{L}\left(u, \lambda \right),
\end{align*}
where $u :=\left(u_1,u_2,\ldots,u_M \right)$, and $\lambda :=\left(\lambda(0),\lambda(1),\ldots,\lambda(T-1)\right)$.
The Lagrangian is the sum of the costs incurred by each individual agent. Hence, given the Lagrange multipliers $\lambda$, the inputs $u_i$ minimizing the Lagrangian can be calculated in a decentralized fashion, with each agent $i$ solving its own problem called Agent $i$'s Problem
\begin{align}
& \min \sum_{t=0}^{T} c_i(x_i(t)) + \lambda(t) u_i(t),\\
&\mbox{subject to } x_i(t+1) = f_i^t(x_i(t),u_i(t)) \notag\\
&\mbox{for }t=0,1,\ldots,T-1.\notag
\end{align}
Each agent $i$ then submits this optimal $u_i(\cdot)$ to the ISO as its bid. This enables the computation of the dual function for each value of $\lambda$.

Note that the sub-gradient with respect to $\lambda$ of the Dual function $D(\lambda)$ is $\left(\sum_i u^k_i(0),\sum_i u^k_i(1),\ldots,\sum_i u^k_i(T-1)\right)$. Since the dual problem of finding the prices $\lambda$ that maximize $D(\lambda)$ is convex, it can be solved via the sub-gradient iteration~\cite{bertsekas,rockafellarconvex,boydlecture}.
\begin{align}
\lambda^{k+1}(t) = \lambda^k(t) \left(1-\alpha_{k}\right) + \alpha_{k} \left(\sum_i u^k_i(t)\right), t\geq 0,
\end{align}
where $k$ is the index which keeps track of the iteration number. The iterations end when the price vector $\lambda(t)$ converges to the optimal value $\lambda^\star(t)$. The resulting solution is optimal for the ISO Problem due to the convexity assumptions. 
\section{Privately Observed Noise}\label{mdp:pi}
Suppose that even though the agents do not observe the private noises of other agents, or know their system dynamics or utility functions, they know the \emph{laws} of the combined noise process, $\mathcal{L}(N(t))$. That is, in the context of the uncertainty tree of Section~\ref{sec:tree}, the agents know the topology of the tree, and the transition probabilities along the edges.

We will now show that under the following convexity assumption, the~\ref{p0} has an optimal solution. 

\begin{assumption}\label{assum1}
The function
\begin{align}\label{conv1}
\sum_t c_i(X_i(t),U_i(t)), i=1,2,\ldots,M,
\end{align}
is convex in the vector $\left\{U_i(t), t=0,1,\ldots,T-1\right\}$ for fixed noise sequence.
\end{assumption}
The algorithm presented is iterative, and composed of Bid-Price updates. The bid submitted by each agent $i$ at time $s$ is a random process that maps the space $\Omega \times \left\{s,s+1,\ldots,T-1\right\}$ to $\mathbb{R}$. This is akin to Arrow's~\cite{arrow} approach of treating each good available at a certain time and place as a separate good. The bid process is adapted to the filtration $\mathcal{F}_t$. At each time $t$, it specifies to the ISO, as a function of the past noise $N(s), s<t$, the amount of electricity that the agent is willing to purchase at time $t$. Only the initial bid, $U_i(s)$ is implemented at time $s$, as in Model Predictive Control. Figure~\ref{flo} summarizes the algorithm.
\begin{theorem}
Algorithm~\ref{a2} solves the~\ref{p0} when the cost functions satisfy Assumption~\ref{assum1}, with each agent $i$ having access only to its private noise $N_i(t)$ and the common noise $N_c(t)$, with the law of the combined noise process, i.e., $\mathcal{L}(N(t))$, being known publicly. 
\end{theorem}
\begin{proof}
Let us first consider the special case where there is only a commonly observed noise $N_c$ and there are no private noises, called the Commonly Observed Noise Problem. 
Suppose for simplicity of visualization the noise process $N(t)$ assumes only finitely many values allowing it to be represented by a tree as in Fig.~\ref{extenform}.

Let us suppose that $x(0)$ is fixed, without loss of generality. Let $p_v$ denote the probability of node $v$ in the uncertainty tree. The depth of the node in the tree indicates time. Every Markov policy, mapping states and time to actions, specifies an action $U(v) := (U_0(v), U_1(v),\ldots, U_M(v))$ satisfying $\sum_i U_i(v) = 0$ for every node $v$ in the tree. This is easily seen by recursion starting at the root which corresponds to the initial time and state of the system, and noting that each node then also indicates the state of the system at that time. Now consider also a more general ``tree policy" that specifies a $U(v) := (U_0(v), U_1(v),\ldots, U_M(v))$ satisfying $\sum_i U_i(v) = 0$ for every node $v$ in the tree. It is more general than a Markov policy since two nodes in the tree at the same depth may correspond to the same state $X(t)$ but a tree policy is allowed to prescribe different actions for them. Hence the class of tree policies also contains an optimal policy. 

For every such tree policy, for every node $v$, there is a unique sequence of actions $U^v := \left\{U(0), U(1),\ldots,U(t)\right\}$ that was taken in the preceding $t$ steps, where $t$ denotes the depth of the node $v$. The state $X(t)$ at time $t$ corresponding to the node $v$ is thereby determined by $(v,u^v)$. The centralized optimization problem can then be written as the following optimization problem,
\begin{align*}
& \min \sum_{i=1}^{M}\sum_v p_v c_i \left(v,U^v\right) \\
& \mbox{ such that } \sum_i U_i(v) = 0,\forall v.
\end{align*}
Note that $c_i(v,u)$ is convex in $u$. Hence this is a convex programming problem with no duality gap. Associating Lagrange multiplier $\lambda(v)$ with the constraint $\sum_i U_i(v) = 0$, and letting $\lambda : = \{\lambda(v) \}$, we obtain,
\begin{align*}  
\mathcal{L}\left(U, \lambda \right):& = \sum_{i=1}^{M}\sum_v p_v\left\{\sum_v c_i(v,u^v) + \lambda (v) U_i(v)\right\}.
\end{align*}
We will call the process $\lambda(v)$ as the ``price process".
Each agent submits a bid for each possible partial realization $v$ of the noise process, while the ISO specifies a price at each $v$. Now the proof parallels the proof in the deterministic case.

This proof extends to the case where there are also private noises. At time $0$, there are no private noises $N_i(-1)$, and so the above proof holds at time $0$. Noting this, it follows that the result also holds at each time $s \geq 1$ since the bid-price iteration is repeated at each such time, and we can simply regard $s$ as the new ``initial" time. 
\end{proof}
\begin{algorithm}
\caption{}
\label{a2}
\begin{algorithmic}
\STATE \textbf{Assumption:}  The law of the combined noise process $\mathcal{L}(N)$ is common knowledge of all agents and ISO. 
 \FOR{ bidding times $s=0$ to $T-1$}
\STATE $k=0$
\STATE \REPEAT
\STATE Each agent $i$ solves the problem
\begin{align}\label{AP1}
\min \mathbb{E}\left\{\sum_{t\geq s} c_i(X_i(t),U_i(t)) + \lambda^{k}(t) U_i(t) \right\},
\tag{Agent i's Problem}
\end{align}
with initial condition $X_i(s)$ for the optimal $\{U^k_i(t),s\leq t\leq T-1\}$, and submits it to ISO.
\STATE ISO declares new price, i.e. 
\begin{align*}
& \lambda^{k+1}(t) = \lambda^k(t) \left(1-\alpha_k\right) + \alpha_k \left(\sum_i U^k_i(t)\right),\\
& \mbox{ for times } t\geq s.
\end{align*}
\STATE $k\to k+1$
\UNTIL{$U^k_i(t)$ converges a.s. to $U^{\star}_i(t),\quad \forall t\geq s$}
\STATE ISO implements $U^{\star}_i(s)$
\ENDFOR
\end{algorithmic}
\end{algorithm}
We note that the assumption of the common knowledge of $\mathcal{L}(N(t))$ can be removed by utilizing the technique of Stochastic Approximation or other learning techniques~\cite{borkarbook,kushner,robbins}, so that the agents can ``learn" the laws $\mathcal{L}(N(t))$.

The above algorithm above is exponentially complex in $T$ due to the number of possible states in the tree, even when there are only two possible values for each $N_i(t)$. In the next section we will see that we can dramatically simplify the algorithm in the LQG context.
\section{The Case of Linear Systems}
This section treats the special case of the~\ref{p0} when the $M$ agents have linear Gaussian dynamics and quadratic costs. The noises of all agents are independent and mean zero. Each agent $i$ has a quadratic cost criterion, i.e., the cost functions $c_i(x_i,u_i) = x_i^\intercal Q_i x_i + u_i^\intercal R_i u_i$ are quadratic, with weighting matrices $Q_i \geq 0$ and $R_i > 0$ . Let us call this the Distributed Constrained LQG (DCLQG) Problem:
\begin{align}\label{p1}
& \min \mathbb{E} \left(\sum_{t=1}^{T}\sum_{i=1}^{N} X^\intercal_i(t)Q_i X_i(t) + U^{\intercal}_i(t) R_i U_i(t)\right)\notag\\
&\mbox{subject to } X_i(t+1) = A_iX_i(t) +B_i U_i(t) + B_i N_i(t),\notag\\
& \qquad t=0,\ldots,T-1,\notag\\
&\mbox{and } \sum_i U_i(t) =0, t=0,\ldots,T-1.
\end{align}
(The case of time-varying systems is analogous to time-invariant systems, and omitted for brevity)
We will assume that the system dynamics given by $(A_i,B_i)$, the cost functions given by $(Q_i,R_i)$, and the observation structure are all private, i.e., none of the agents have knowledge of the system parameters or the costs of the other agents, and that the state process $X_i$ is observed only by the agent $i$. 

We will derive an Iterative Bidding Scheme which is much simpler than the algorithm proposed in Section~\ref{mdp:pi} in the following critical aspect: The bid function submitted at time $t$ specifying the quantity of electricity that agent $i$ is willing to purchase at times $t,t+1,\ldots,T-1$ does not depend on the outcomes of noise sequence $N(s),s\geq t$. It is simply a vector $(u_i(t),u_i(t+1),\ldots,u_i(T-1))$ comprising of $T-t+1$ entries.
This is a drastic reduction in complexity of the bidding scheme. At each time $t$, the following iteration takes place: Each agent bids a vector of future purchases corresponding to a deterministic certainty equivalent system, in response to prices announced by the ISO for future power, and the ISO updates the prices in return, until convergence.
\begin{definition}[Certainty Equivalence]
A stochastic control problem is said to possess the property of certainty equivalence if the optimal policy for the stochastic control problem coincides with the optimal policy for the corresponding deterministic control problem in which the noise is absent. 
\end{definition}
\begin{theorem}
The following bidding scheme achieves optimality for the ISO Problem with LQG agents. At each time $s$, in respose to the $k$-th iterate of the price sequence $(\lambda^k(s),\lambda^k(s+1),\ldots,\lambda^k(T))$, agent $i$ announces the optimal open loop sequence $(u^k_i(s), u^k_i(s+1),\ldots,u^k_i(T))$ for the deterministic LQ problem:
\begin{align*}
&\min \sum_{t\geq s} x_i^\intercal (t) Q_i x_i(t) + u_i^\intercal (t) R_i u_i(t)+ \lambda^k(t)u_i(t)\\
&\quad \mbox{s.t. } x_i(t+1) = A_ix_i(t) + B_iu_i(t) \\
&\mbox{ for } t=s, s+1,\ldots,T-1.
\end{align*}
In response, the ISO adjusts the prices according to:
$\lambda^{k+1}(t) = \lambda^k(t) \left(1-\alpha_k\right) + \alpha_k \left(\sum_i u^k_i(t)\right), t\geq s$. This process is iterated till it converges to $(u^{\star}_i(s),u^{\star}_i(s+1),\ldots,u^{\star}_i(T-1))$
and $(\lambda^{\star}(s),\lambda^{\star}(s+1),\ldots,\lambda^{\star}(T-1))$.
At time $s$, the price is set at $\lambda^{\star}(s)$ and agent $i$
applies the input $u^{\star}_i(s)$.
\end{theorem}
\begin{algorithm}
\caption{}
\label{a3}
\begin{algorithmic}
 \FOR{bidding times $s=0$ to $T-1$}
\STATE $k=0$
\STATE  Initialize $\lambda^k (t),t\geq s$ to some arbitrary value.
 \STATE \REPEAT
 \STATE Each agent $i$ solves the problem 
 \begin{align}
 \min \sum_{t\geq s}^{T} x_i^\intercal (t) Q_i x_i(t) +u_i^\intercal (t) R_i u_i(t)+ \lambda^k(t)u_i(t)
 \end{align}
and submits the optimal values, denoted $u_i^k(t)$ for $t\geq s$ to the ISO.
 \STATE ISO updates the prices,
 \begin{align*}
& \lambda^{k+1}(t) = \lambda^k(t) \left(1-\alpha_k\right) + \alpha_k \left(\sum_i u^k_i(t)\right), \\
& \mbox{ for times }t\geq s.
\end{align*}
 Increment $k$ by $1$
\UNTIL{$u^k_i(t)$ converges to $u^{\star}_i(t)$}
\STATE implement $u^\star(s)$
\ENDFOR
 \end{algorithmic}
 \end{algorithm}
 The key to showing the existence of such a simple bidding scheme lies in utilizing the certainty equivalence property of LQG systems~\cite{kumar}.
\begin{proof}
Let 
\begin{align*}
& x:=(x_1,x_2,\ldots,x_M), u:=(u_1,u_2,\ldots,u_M),\\
& A := diag(A_1, A_2,\ldots,A_M), B:=diag(B_1,B_2,\ldots,B_M),\\ 
& Q=diag(Q_1,Q_2,\ldots,Q_M), R=diag(R_1,R_2,\ldots,R_M),
\end{align*}
 and consider the following deterministic linear, quadratic regulator (LQR) problem with no noise, but featuring the energy balance constraint,
\begin{align}
& \min \sum_{t=0}^{T} x^\intercal (t)Qx(t) + u^{\intercal}(t) R u(t)\\
&\mbox{subject to } x(t+1) = Ax(t) +B u(t),\label{s:eq}\\
& \sum_{i=1}^{M} u_i(t) = 0 \mbox{ for } t=0,\ldots,T-1.\notag
\end{align}
Since the state is affine in $u$, the cost is convex in $u$. Hence this centralized problem can be solved by the Bid-Price iteration between the agents and the ISO as shown for the deterministic problem. In particular, at time $0$, the end result of the scheme is the optimal action $u(0)$. This is arrived at by the ISO announcing a sequence of prices for all future times and the agents bidding their consumptions/generation sequences at all future times. 
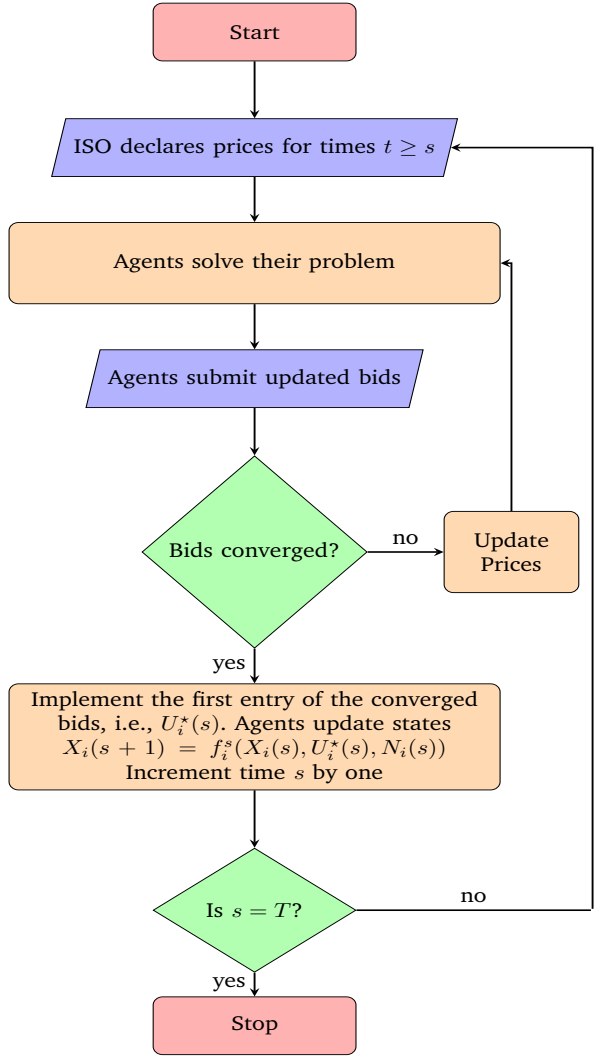
\begin{figure}[h]\centering\resizebox{8cm}{14cm}{ 
\begin{tikzpicture}[node distance=2cm]
\tikzset{trapezium stretches=true}
\node (start) [startstop] {Start};
\node (in1) [io, below of=start] {ISO declares prices for times $t\geq s$ };
\node (pro1) [process, below of=in1] {Agents solve their problem}; 
\node (in2) [io, below of=pro1] {Agents submit updated bids}; 
\node (dec1) [decision, below of=in2,yshift=-1cm] {Bids converged?};
\node (pro2) [upd, right of=dec1,xshift=1.8cm] {Update Prices};
\node (pro3) [process, below of=dec1,yshift = -1.2cm] {Implement the first entry of the converged bids, i.e., $U^{\star}_i(s)$. Agents update states $X_i(s+1)=f_i^s(X_i(s),U_i^{\star}(s),N_i(s))$ Increment time $s$ by one};
\node (dec2) [decision, below of=pro3,yshift = -1cm] {Is $s=T$?};
\node (stop) [startstop, below of=dec2] {Stop};
\node (dummy) [circle,radius = .0pt,inner sep=0pt,right of = dec2,xshift = 3cm]{};
\draw [arrow] (start)--(in1);
\draw [arrow] (in1)--(pro1);
\draw [arrow] (pro1)--(in2);
\draw [arrow] (in2)--(dec1);
\draw [arrow] (dec1)-- node[anchor=east] {yes} (pro3);
\draw [arrow] (pro3) -- (dec2);
\draw [arrow] (dec2) --node[anchor=east] {yes}  (stop);
\draw [arrow] (pro2) |- (pro1);
\draw [arrow] (dec1) -- node[anchor=south] {no} (pro2);
\draw [thick] (dec2) -- node[anchor=south]{no} (dummy);
\draw [arrow] (dummy) |- (in1);
\end{tikzpicture}
}
\caption{Decision flow in Algorithm~\ref{a2}.}
\label{flo}
\end{figure}
Now note that due to the energy balance at each time, agent $M$ is forced to choose $u_M(t) = - \sum_i^{M-1}u_i(t)$ for all $t$. Hence one can substitute this value for $u_M(t)$ and obtain a standard LQR problem where there is no separate energy balance constraint. For this reduced and standard deterministic linear quadratic regulator problem, the optimal solution is given by linear feedback $u(0) = \Gamma(0) x(0)$, where $\Gamma(\cdot)$ is the optimal feedback gain.

Now consider the corresponding reduced stochastic LQG problem where there is white Gaussian noise in the state equations~\eqref{s:eq}. By Certainty Equivalence~\cite{kumar}, the same feedback law as in the deterministic reduced LQR problem is also optimal. In particular, in state $x(0)$ at time $0$, $u(0)=\Gamma(0)x(0)$ continues to be optimal. Now, in our proposed bidding scheme for the LQG problem, each agent bids on the basis of a private deterministic system for itself. Hence it leads to the same Bid-Price iteration result at time $0$. Hence it arrives at the same $u(0)$, which however is also optimal for the stochastic LQG problem. 

Thus we see that the Bid-Price iteration scheme determines the optimal actions for the agents at time $0$. Now our scheme for the LQG problem repeats such a Bid-Price scheme iteration at each time $t$. Each $x(t)$ can be regarded as an initial state for the system started at time $t$, and the same argument as above shows that the actions $u(t)$ that it results in for the agents at all times $t$ are also optimal.
\end{proof}
We note the following important aspects of the proposed algorithm. The critical feature that there is an iteration of bids at each time $t$ is important. Also important is that at each stage it is the future sequence of prices that is iterated. 

It should be noted that the alternative of announcing a ``bid curve" of price vs. generation for a single time $t$ does not work in the dynamic case. The reason is that the current optimal generation depends on future prices, so iteration of price at only one time is not sufficient to ensure optimal decisions  when agents are dynamic systems.

\section{Concluding Remarks}
We have posed the~\ref{p0} of maximizing the total utility/minimizing the total operating cost of the electricity grid, while not allowing the revelations of the dynamics, states or utilities of the agent. It is more complex than a decentralized stochastic control problem due to the non-revelation feature. We have shown that when the agents are LQG systems the problem admits a simple solution utilizing iterative bidding schemes, and attains the same performance as that of an optimal centralized control policy. Under the proposed policy, the sufficient statistics are vastly simplified, and each agent $i$ needs to only keep track of its present state $X_i(t)$. This is in contrast to the general case of decentralized stochastic control, in which the agents need to keep track of the entire history in order to implement an optimal policy, which is in any case generally intractable to compute. Thus not only is our Algorithm decentralized, and easy to implement, but it also leads to a large amount of data reduction. We further note that our Algorithm is privacy preserving.
\bibliographystyle{IEEEtran}
\bibliography{../GTS/combinedbib}

\end{document}